\documentclass[letterpaper, 10 pt, conference]{ieeeconf}

\newif\ifachproofinappendix
\achproofinappendixtrue

\IEEEoverridecommandlockouts
\usepackage[cmex10]{amsmath}
\usepackage{cite}

\title{\LARGE \bf
Rate-cost Trade-offs in $H_\infty$ Control with Initial State Uncertainty
}

\usepackage{hyperref}
\usepackage{bm}
\usepackage{bbm}
\usepackage{url}

\usepackage{amsthm}
\usepackage{amscd}
\usepackage{amssymb}

\usepackage{mathtools}
\usepackage{mathrsfs}

\usepackage{tikz}
\usepackage{xcolor}
\usepackage{graphicx}
\usepackage{svg}

\usepackage{nicefrac}
\usepackage{microtype}

\usepackage{caption}
\usepackage{subcaption}

\usepackage{times}
\usepackage{xspace}
\usepackage{tabularx}

\usepackage[capitalize,noabbrev]{cleveref}

\usepackage{booktabs}

\newcommand{\ic}{x_0 \in \Omega_0}

\newcommand{\lb}{\alpha}
\newcommand{\ub}{\beta}

\newtheorem{theorem}{Theorem}[section]

\newtheorem{definition}[theorem]{Definition}

\newtheorem{problem}[theorem]{Problem}

\newtheorem{remark}[theorem]{Remark}

\newcommand{\ie}{\textit{i.e.}}

\newcommand{\vol}{\text{Vol}}

\newcommand{\rcf}{R({\gamma})}

\newcommand{\beq}[1]{\begin{equation}\label{eq:#1}}
\newcommand{\eeq}{\end{equation}}

\newcommand{\defeq}{\coloneqq}

\newcommand{\inn}{\!\in\!}

\newcommand{\norm}[2][\text{}]{\Vert{#2}\Vert_{#1}}

\newcommand{\Z}{\mathbb{Z}}

\newcommand{\R}{\mathbb{R}}

\newcommand{\F}{\mathcal{F}}

\newcommand{\Hinf}{\mathit{H}_\infty}

\newcommand{\Bg}{B_{\gamma}}

\author{Vikrant Malik, Victoria Kostina and Babak Hassibi\thanks{Vikrant Malik, Victoria Kostina
and Babak Hassibi are with the Department of Electrical Engineering, California Institute of
Technology,
        Pasadena, California
        {\tt\small \{vmalik, vkostina, hassibi\}@caltech.edu}}
\thanks{\copyright\ 2026 IEEE. Personal use of this material is permitted.}}

\begin{document}

\maketitle
\thispagestyle{empty}
\pagestyle{empty}

\begin{abstract}
We consider the problem of $\Hinf$ control in the presence of a digital communication channel
between the observer and the controller and investigate the fundamental trade-off between the
required communication data rate and $\Hinf$ performance.
The system follows linear dynamics and has no additive internal noise, so that the uncertainty is
confined to the initial state. For the scalar system, using a deterministic time-zero
covering argument, we establish a lower bound on the required data rate for a given $\Hinf$ cost
and demonstrate why logarithmic quantization (Elia, 2000) is a natural choice.
We further develop an achievability scheme and show that, for the scalar system, its rate
matches the converse asymptotically in the high-data-rate limit.
\end{abstract}

\section{Introduction}

\subsection{Literature Context}
Control systems under communication constraints are ubiquitous in modern technologies, including
distributed robotics, networked autonomous systems, and smart grids. These systems tackle the dual
challenges of maintaining precise control while operating under communication uncertainties, such
as bandwidth restrictions or latency. For example, drones rely on real-time sensor-controller
communication, where a rate-constrained or delayed data transfer can severely impact safety and
performance. Similarly, industrial automation systems often operate within shared communication
networks, requiring robust mechanisms to ensure stability and efficiency despite channel
limitations. This interplay between control and communication highlights the need to understand how
communication constraints impact control objectives at a fundamental level.

The study of control under communication constraints has a well-established history, with results
exploring the trade-offs between control performance and required communication characteristics.
The first works that explored the connection between control theory and communication constraints
appeared in \cite{baillieul1999feedback, wong1999systems}. These works introduced the concept of
stabilizability in scalar systems under a rate constraint between the observer and the controller.
They showed that stabilization requires a minimum data rate exceeding \( \log_2 |A| \) bits per
sample, where \( A \) represents the system parameter. These results were later extended to vector
systems in \cite{tatikonda2004control}, where the authors showed that
\begin{align}
    R > \sum_{{|\lambda_i(A)|}> 1} \log_2 \left|\lambda_i(A)\right|
    \label{eq:drt}
\end{align}
is a necessary condition on the communication rate to achieve stabilizability. These results were
extended by \cite{nair2007feedback}, who used a volume division argument to show that the same
condition on rate holds in the unbounded noise setting.
However, fixed-rate quantizers cannot account for arbitrarily large noise realizations, prompting
the development of adaptive quantization schemes that dynamically adjust quantization intervals
based on system state \cite{nair2004stabilizability, brockett2000quantized, yuksel2013jointly,
kostina2021exact}. Due to this, variable rate \cite{silva2010framework, kostina2019rate}, and even
infinite rate quantizers \cite{elia2001stabilization, elia2000design, fu2005sector} have also been
explored. In \cite{elia2001stabilization}, the authors showed that the optimal density of an
infinite rate quantizer to quadratically stabilize a system is logarithmic. High density
quantization is another active field of research \cite{lookabaugh1989high, korten2007high,
linder1999high}. In this context, companders have been used to generate high-density quantizers.
For example, in \cite{linder1999high}, the authors study a companding scheme for entropy-coded
vector quantization. The core idea behind a compander is to uniformly quantize a compressed version
of a stochastic source. However, that framework is formulated for stochastic sources
and average distortion, and therefore does not directly provide the deterministic worst-case
guarantee pursued here.

Linear quadratic Gaussian (LQG) control with communication constraints was studied in
\cite{tatikonda2004stochastic, charalambous2008lqg, stavrou2021sequential, kostina2019rate}, where
the minimum achievable LQG cost was linked to the Gaussian causal rate-distortion function
\cite{gorbunov1974prognostic} and directed mutual information constraints.
Furthermore, \cite{han2024coded} investigated squared stability of linear systems over vector
Gaussian channels.
Stabilization results for nonlinear systems have also been obtained \cite{Nair2004, Liberzon2005}.
The authors in \cite{Nair2004} introduced the concept of topological feedback entropy to quantify
the fundamental data rate required for the stabilization of nonlinear plants. Similarly,
\cite{Liberzon2005} demonstrated how input-to-state stability (ISS) properties can be leveraged to
counteract state estimation errors and achieve global asymptotic stabilization of nonlinear systems
with limited information.

The $\Hinf$ cost penalizes the worst-case amplification of disturbances by a system, capturing how
large the output can get relative to the most adverse bounded input. Targeting a specific $\Hinf$
cost is harder than mere stabilizability considered in \cite{tatikonda2004control,
nair2004stabilizability, nair2007feedback, brockett2000quantized}.
The first work that addressed quantized $\Hinf$ control \cite{elia2000design} looked at the problem
setting in the infinite data rate regime and showed that the optimal structure for the quantizer is
one with logarithmic density. The authors consider the problem with the additional constraint that
the state is quadratically stabilizable. Their work was extended in \cite{fu2005sector} where a
certain sector bound approach is used to show that the optimal quantizer is logarithmic. A
finite-level logarithmic quantizer with dynamic scaling was proposed in \cite{fu2009finite},
achieving asymptotic stabilization with a moderate number of quantization levels.
A related but distinct problem was studied in \cite{niu2014control}, where an adaptive dynamic
quantizer is designed to preserve a predesigned stochastic $\Hinf$ disturbance-attenuation level
over a quantized channel with packet losses.
\subsection{Contributions}
In this paper, we elucidate the fundamental trade-off between the data rate and the achievable
$\Hinf$ cost in control of a linear system.
Unlike \cite{tatikonda2004control, nair2004stabilizability, nair2007feedback,
brockett2000quantized}, we focus on the $\Hinf$ cost, and unlike \cite{elia2000design,
fu2005sector}, we consider finite rate quantization. We study this problem in the noiseless setting
where the only uncertainty is the initial state.
Our converse uses a deterministic time-zero covering argument
in $z=\log|x|$. Unlike entropy methods, it needs no source law and unlike
invariant-set or volume-growth methods, it obtains a $\gamma$-dependent
rate constraint without propagating uncertainty. This converse motivates
a vector achievability scheme whose scalar specialization attains its
leading high-rate term.
Together, our achievability and converse results yield a refinement of the data-rate theorem
\eqref{eq:drt}, in which the objective is not merely to ensure stability, but to guarantee a
prescribed $\Hinf$ gain. This parallels the approach of \cite{kostina2019rate}, which refines the
same theorem by requiring a target LQG cost rather than simply the absence of instability. Our
refinement, however, is restricted to the scalar noiseless setting.
Finally, numerical experiments illustrate the theoretical bounds and show improved performance
over the comparison schemes in the scalar example considered.

\subsection{Paper Organization}
The remainder of this paper is organized as follows: Section \ref{sec:prelim} contains the
notations used in the paper. We also touch upon the problem of $\Hinf$ control without a
communication constraint to better set up the rest of the paper. In Section \ref{sec:converse}, we
provide a lower bound on the minimum data rate needed to achieve a certain $\Hinf$ performance
using the argument developed there. In Section \ref{sec:ach}, we introduce an
achievability scheme motivated by the converse. Section \ref{sec:num} presents numerical
simulations that validate the theoretical results.

\section{Preliminaries} \label{sec:prelim}
\subsection{Notations}
We write $\R$ for the real numbers and $\Z_+$ for the positive
integers.  For a matrix $A$, $A^*$ denotes its transpose.  The symbol
$\norm{\cdot}$ denotes the Euclidean norm.
The unit sphere is $S^{d-1}\defeq\{v\in\R^d:\norm{v}=1\}$,
$\angle(v,w)\in[0,\pi]$ is the angle between nonzero vectors, and
$[a]_+\defeq\max\{a,0\}$.

\subsection{Dynamical System}
Consider a discrete-time, linear time-invariant (LTI) dynamical system expressed as a state-space
model given by:
\begin{align}
   x_{t+1} &= A x_{t} + B_u u_{t} , \quad \ic. \label{eq:system}
\end{align}
Here, $x_{t} \inn \R^{d}$ is the \emph{state}, $u_t \inn \R^{d_u}$ is the \emph{control input}. The
state-space parameters $(A,B_u)$ are known and stabilizable. The initial state $x_0$ is unknown to
the controller but is known to be in a set $\Omega_0 \subseteq \R^{d}$ where,
\begin{align}
    \Omega_0 \defeq \{x \in \R^d \mid \lb \leq \norm{x} \leq \ub\},
    \qquad 0<\lb<\ub.
\end{align}
We assume that $A$ is diagonalizable with real eigenvalues,
and choose $\Phi_A$ such that
$\Phi_AA\Phi_A^{-1}=\operatorname{diag}(\lambda_1(A),\ldots,\lambda_d(A))$.

\subsection{Communication Channel}
\label{sec:comms}
The communication channel is a noiseless bit pipe with a finite alphabet
$\Sigma$ and rate $R\defeq\log_2|\Sigma|$ bits per sample. At each time $t$,
the observer transmits a symbol $\sigma_t\in\Sigma$, which the controller
receives without corruption. The symbol $\sigma_t$ is generated via a causal mechanism that
depends on the state observations up to time $t$, \ie, $x_0, \dots, x_t$.

\subsection{\texorpdfstring{$\Hinf$}{text} control objective (Noiseless Case)}
Let $\Pi_0,Q,W\succ0$ and $P_f\succeq0$ be fixed weighting
matrices. For a horizon $N$, define the weighted state-and-control cost
\begin{align}
J_N^{P_f}(x_0)
\defeq x_{N+1}^*P_fx_{N+1}
+\sum_{t=0}^N\left(u_t^*Qu_t+x_t^*Wx_t\right).
\end{align}
The $\Hinf$ control problem under a communication constraint is
then formulated as follows.
\begin{problem}\label{prob:quant_hinf}
Given $\gamma>0$ and channel rate $R$, find one causal coding-and-control strategy
$\{f_t,g_t\}_{t\geq0}$, where
$\sigma_t=f_t(x_0,\ldots,x_t)$ and
$u_t=g_t(\sigma_0,\ldots,\sigma_t)$, such that
\begin{align}
    \limsup_{N\to\infty}\ \sup_{x_0\in\Omega_0}
    \frac{J_N^{P_f}(x_0)}{x_0^*\Pi_0^{-1}x_0}
    <\gamma^2.
    \label{eq:prob}
\end{align}
\end{problem}
The ratio in~\eqref{eq:prob} is the energy gain from the uncertain
initial state to the weighted state-and-control trajectory.
Thus one horizon-independent strategy
must meet an asymptotic worst-case induced-$\ell_2$ gain bound. This is a
noiseless initial-condition specialization of the standard
control-theoretic $H^\infty$ disturbance-attenuation problem
\cite{francis_hinf_1987,doyle_state-space_1988,khargonekar_transients_1991,
farhood_uncertain_ic_2008,farhood_lpv_ic_2021}.
To gain an insight into Problem~\ref{prob:quant_hinf}, we first
look at the problem without any communication constraint. Let $P\succeq0$
be the stabilizing solution of the algebraic Riccati equation, and define
$R_e\defeq Q+B_u^*PB_u$ and $K\defeq R_e^{-1}B_u^*PA$, so that
\begin{align}
    P&=A^*PA+W-K^*R_eK.
    \label{eq:stationary_riccati}
\end{align}
On rearranging the cost in~\eqref{eq:prob} and using a
completion of squares argument, we can simplify this expression. Using
\cite[Th.~9.2.1]{blackbook}, we have, for every $N$,
\begin{align}
    J_N^{P_f}(x_0)
    &=x_0^*Px_0+\sum_{t=0}^{N}
    (u_t+Kx_t)^*R_e(u_t+Kx_t)\nonumber\\
    &\quad+x_{N+1}^*(P_f-P)x_{N+1}.
    \label{eq:stationary_cost_identity}
\end{align}
Let $F\defeq A-B_uK$ and $C_K^*C_K=W+K^*QK$. Since
$P=F^*PF+C_K^*C_K$, Parseval's identity gives the frequency-domain form
\begin{align*}
    x_0^*Px_0
    =\frac{1}{2\pi}\int_{-\pi}^{\pi}
    \left\|C_K(I-e^{-\jmath\omega}F)^{-1}x_0\right\|^2\,d\omega.
\end{align*}
Thus, this is initial-condition transient energy.
Define
\begin{align}
    M&\defeq K^*R_eK,
    &H_\gamma&\defeq\gamma^2\Pi_0^{-1}-P.
    \label{eq:stationary_weights}
\end{align}
Feasibility and $W,Q\succ0$ imply $x,u\in\ell_2$ for each
$x_0\in\Omega_0$; hence
$x_{N+1}\to0$, the terminal term in
\eqref{eq:stationary_cost_identity} vanishes, and its nonnegative partial
sum converges. Taking $N\to\infty$, every feasible strategy therefore
satisfies
\begin{align}
    \sum_{t=0}^{\infty}(u_t+Kx_t)^*R_e(u_t+Kx_t)
    <x_0^*H_\gamma x_0,
    \qquad \forall x_0\in\Omega_0.
    \label{eq:J_scalar}
\end{align}
The controller $u_t=-Kx_t$ is referred to as the central
controller. Throughout our discussion, we assume that $H_\gamma\succ0$.
Equation~\eqref{eq:J_scalar} says that the cumulative deviation from the
central controller is bounded above by $x_0^*H_\gamma x_0$. Since $\Omega_0$
is bounded, the cumulative deviation is bounded. We can now ask the
question, how much rate is required to achieve the $\Hinf$ cost $\gamma$.

\begin{definition}
Consider Problem~\ref{prob:quant_hinf}. Under a data rate
constraint between the observer and controller, the control inputs are a
function of the received symbols $\sigma_t$ defined in
Section~\ref{sec:comms}. A valid control strategy
$u_t=\F_t(\sigma_0,\sigma_1,\ldots,\sigma_t)$ for all $t\geq0$ is one
that guarantees~\eqref{eq:prob}. The set of all such control strategies for
a given rate $R$ is denoted by $\Gamma_{\gamma,R}$.
\end{definition}

\begin{definition}
The $\Hinf$ rate-cost function of the system
\eqref{eq:system} is defined as the minimum rate $\rcf$ such that
$\Gamma_{\gamma,\rcf}$ is a non-empty set. If no such rate exists, we set
$\rcf=+\infty$.
\label{def:rcf}
\end{definition}

\section{Converse} \label{sec:converse}

In this section, we provide a fundamental lower bound on the rate-cost
function in Definition~\ref{def:rcf} for a scalar system using a
modified volume-division argument
\cite{tatikonda2004control}. We show that the required quantization error depends on
the magnitude of the initial state, so the corresponding quantization
regions do not have constant Euclidean length.  We account for this state
dependence using a logarithmic quantizer, under which all admissible
relative-error regions have the same length.  The main result is stated
below.

\begin{theorem}\label{th:main}
Suppose that \(d=d_u=1\), \(0<\lb<\ub\), and \(K\neq0\).  Define
\begin{align}
    \Bg
    &\defeq
    \frac{\gamma^2\Pi_0^{-1}-P}{K^2R_e}.
    \label{eq:def_Bgamma_converse}
\end{align}
If \(0<\Bg<1\), then
the rate-cost function satisfies
the lower bound
\begin{align}
    \rcf
    &>
    \log_2\left(
    \frac{\log(\ub/\lb)}
    {\operatorname{tanh}^{-1}\left(\sqrt{\Bg}\right)}
    \right).
    \label{eq:main_th_c}
\end{align}
\end{theorem}

\begin{proof}
Fix any valid strategy with alphabet \(\Sigma\).  At time \(t=0\), the
controller has received one symbol and can therefore select at most
\(|\Sigma|\) distinct control values. For every such value \(u_0\), define
an associated point
\begin{align}
    \widehat{x}_0\defeq-\frac{u_0}{K}.
    \label{eq:control_induced_point}
\end{align}
Here, \(\widehat{x}_0\) is only an auxiliary point induced by \(u_0\), not
a state reconstruction.
For every initial state
\(x_0\in\Omega_0=[-\ub,-\lb]\cup[\lb,\ub]\) whose transmitted symbol
induces the control value \(u_0\), considering only the \(t=0\) term on
the left-hand side of \eqref{eq:J_scalar} gives the necessary condition
\begin{align}
    K^2R_e(x_0-\widehat{x}_0)^2
    <
    x_0^2\left(\gamma^2\Pi_0^{-1}-P\right).
\end{align}
Using \eqref{eq:def_Bgamma_converse}, this becomes
\begin{align}
    |x_0-\widehat{x}_0|<\sqrt{\Bg}|x_0|,
    \qquad \text{for every such }x_0.
    \label{eq:tilde_error_first}
\end{align}
Thus, every valid strategy induces a finite set of auxiliary points that satisfy
\eqref{eq:tilde_error_first}.
This relative-error geometry is also related to the
Weber--Fechner law, according to which the logarithmic scale is optimal
for maximal relative error \cite{portugal2011weber}.

We can now calculate the number of control-induced points ($\widehat{x}_0$) required to
cover \(\Omega_0\) using a volume-division argument.
Define the logarithmic volume of a measurable set
\(E\subset\mathbb{R}\setminus\{0\}\) by
\begin{align}
    \vol_{\log}(E)
    \defeq
    \int_E\frac{dx}{|x|}.
    \label{eq:log_volume}
\end{align}
On either sign component, \eqref{eq:log_volume} is ordinary
one-dimensional volume after the change of coordinates
\(z=\log|x|\).  In particular,
\begin{align}
    \vol_{\log}([\lb,\ub])
    =
    \vol_{\log}([-\ub,-\lb])
    =
    \log\left(\frac{\ub}{\lb}\right).
    \label{eq:initial_log_volume}
\end{align}
For a given \(\widehat{x}_0\), define its admissible
region as
\begin{align}
    \mathcal C(\widehat{x}_0)
    \defeq
    \left\{x\in\mathbb{R}\setminus\{0\}:
    |x-\widehat{x}_0|<\sqrt{\Bg}|x|\right\}.
    \label{eq:admissible_cell}
\end{align}
Because \(0<\Bg<1\), every point in \(\mathcal C(\widehat{x}_0)\) has the
same sign as \(\widehat{x}_0\), and
\(\mathcal C(0)=\varnothing\).  Consider first \(\widehat{x}_0>0\).
For \(x>0\), \eqref{eq:tilde_error_first} is equivalent to
\begin{align}
    \frac{\widehat{x}_0}{1+\sqrt{\Bg}}
    <
    x
    <
    \frac{\widehat{x}_0}{1-\sqrt{\Bg}}.
    \label{eq:positive_relative_cell}
\end{align}
This implies that every nonempty admissible region has the same logarithmic volume,
\begin{align}
    \vol_{\log}\bigl(\mathcal C(\widehat{x}_0)\bigr)
    &=
    \int_{\frac{|\widehat{x}_0|}{1+\sqrt{\Bg}}}^{
          \frac{|\widehat{x}_0|}{1-\sqrt{\Bg}}}
    \frac{dx}{x} = \log \left(\frac{1+\sqrt{\Bg}}{1-\sqrt{\Bg}}\right) \nonumber\\
    &=
    2\operatorname{tanh}^{-1}\left(\sqrt{\Bg}\right),
    \qquad \widehat{x}_0\neq0.
    \label{eq:relative_cell_log_volume}
\end{align}

Now, consider a valid strategy $u_t \in \Gamma_{\gamma,R}$ and let \(N_+\) denote its number of
distinct
positive control-induced points, written as
\(\widehat{x}_{0,1},\ldots,\widehat{x}_{0,N_+}\).  Their respective regions
\(\mathcal C(\widehat{x}_{0,i})\) must cover \([\lb,\ub]\).
After the change of
coordinates \(z=\log x\), their union is an open set containing the compact
interval \([\log\lb,\log\ub]\). Hence it contains
\([\log\lb-\varepsilon,\log\ub+\varepsilon]\) for some
\(\varepsilon>0\).
Subadditivity and \eqref{eq:relative_cell_log_volume} give
\begin{align}
    2N_+\!\operatorname{tanh}^{-1}\left(\sqrt{\Bg}\right)
    &=
    \sum_{j=1}^{N_+}
    \vol_{\log}\bigl(\mathcal C(\widehat{x}_{0,j})\bigr)
    \nonumber\\
    &\geq
    \vol_{\log}\left(
    \bigcup_{j=1}^{N_+}\mathcal C(\widehat{x}_{0,j})\right)
    >
    \log\left(\frac{\ub}{\lb}\right).
    \label{eq:positive_count}
\end{align}
Applying the same volume-division argument to the negative component
\([-\ub,-\lb]\), using \(z=\log|x|\), gives
\begin{align}
    2N_-\!\operatorname{tanh}^{-1}\left(\sqrt{\Bg}\right)
    >
    \log\left(\frac{\ub}{\lb}\right),
    \label{eq:negative_count}
\end{align}
where \(N_-\) is the number of distinct negative control-induced points.
Note that no control-induced point can serve both components, and the zero point
serves neither.
Since \(N_+\) and \(N_-\) are integers,
\eqref{eq:positive_count}--\eqref{eq:negative_count} imply
\begin{align}
    |\Sigma|
    \geq N_++N_-
    &\geq
    2\left(
    \left\lfloor
    \frac{\log(\ub/\lb)}
    {2\operatorname{tanh}^{-1}\left(\sqrt{\Bg}\right)}
    \right\rfloor+1
    \right) \nonumber\\
    &> \frac{\log(\ub/\lb)}
    {\operatorname{tanh}^{-1}\left(\sqrt{\Bg}\right)}
\end{align}
For every valid rate \(R\),
\(|\Sigma|=2^R\), so taking logarithms on both sides yields \eqref{eq:main_th_c}.
\end{proof}

\begin{remark}
When \(|A|>1\), combining \eqref{eq:main_th_c} with the stabilizability lower bound gives
\begin{align}
    \rcf
    >
    \max\left\{
    \log_2\left(
    \frac{\log(\ub/\lb)}
    {\operatorname{tanh}^{-1}\left(\sqrt{\Bg}\right)}
    \right),
    \log_2|A|
    \right\}.
    \label{eq:combined_converse}
\end{align}
The first term in \eqref{eq:combined_converse} is the
performance-dependent rate penalty established in Theorem~\ref{th:main}, and
the second is the classical rate required for stabilization.  Thus the
first term is the active constraint when a stringent \(H_\infty\)
performance is required, whereas \(\log_2|A|\) is the rate floor when the
performance requirement is sufficiently relaxed.
\end{remark}
To make the high-rate behavior explicit, let $\gamma_0 \defeq\sqrt{\Pi_0P}$
denote the minimum achievable \(\gamma\) under full information
\cite{blackbook}.
Then
\begin{align}
    \Bg
    =
    \frac{\gamma^2-\gamma_0^2}
    {\Pi_0K^2R_e}.
    \label{eq:Bg_performance_gap}
\end{align}
Here, the high-data-rate regime means the fixed-system limit
$\gamma\downarrow\gamma_0$, equivalently $\Bg\downarrow0$, for which the
required rate diverges.
As \(\gamma\downarrow\gamma_0\), the performance-dependent term
in~\eqref{eq:combined_converse} scales as
\(-\frac{1}{2}\log_2{(\gamma^2-\gamma_0^2)}+O(1)\); hence the required
rate diverges.
Conversely, as \(\gamma \rightarrow \infty\),
\(\Bg\) increases and the performance-dependent term decreases until the
stabilization term becomes active.  Equation~\eqref{eq:combined_converse}
is stated in the regime \(0<\Bg<1\) of Theorem~\ref{th:main}. Independently,
the stabilization bound
\begin{align}
    \rcf>\log_2|A|
    \label{eq:loga}
\end{align}
continues to hold for every feasible \(\gamma\).

\section{Achievability} \label{sec:ach}

In this section, we propose a logarithmic quantization scheme for
Problem~\ref{prob:quant_hinf} that, in the scalar high-rate regime, attains
the leading term of the converse bound~\eqref{eq:main_th_c}.  The
construction is stated for the vector system.
\begin{definition}[Time-zero polar logarithmic quantizer]
\label{def:log}
Let $L\defeq|\Sigma|=2^R$.  Choose $N_r\in\Z_+$ radial layers
and an angular codebook
$\mathcal V=\{v_1,\ldots,v_{N_a}\}\subset S^{d-1}$.  For $d\geq2$, choose
$0<\Theta<\pi/2$ such that $\mathcal V$ is a geodesic $\Theta$-cover:
\begin{align}
    \min_{1\leq j\leq N_a}\angle(u,v_j)\leq\Theta,
    \qquad u\in S^{d-1}.
\end{align}
For $d=1$, take $\Theta=0$ and
$\mathcal V=S^0=\{-1,+1\}$, so that $N_a=2$. Let $T\defeq\log(\ub/\lb)$.  The radial component is a
logarithmic
quantizer with boundaries
\begin{align}
    \rho_i\defeq\lb\left(\frac{\ub}{\lb}\right)^{i/N_r},
    \qquad i=0,\ldots,N_r.
\end{align}
The radial shells are $[\rho_{i-1},\rho_i)$ for $i<N_r$ and
$[\rho_{N_r-1},\rho_{N_r}]$ for $i=N_r$.  Given $x\in\Omega_0$, the
encoder selects its radial shell and a $v_j$ satisfying the covering
condition, with a fixed rule for resolving ties, and transmits $s=(i,j)$.
Let $C_s$ denote the set of states assigned to this symbol. Define
\begin{align}
    \delta_0^2
    &\defeq \sin^2(\Theta)
    +\tanh^2\left(\frac{T}{2N_r}\right)\cos^2(\Theta).
    \label{eq:polar_delta}
\end{align}
The reconstruction radius for shell $i$ is
\begin{align}
    \bar\rho_i
    \defeq\sqrt{1-\delta_0^2}\sqrt{\rho_{i-1}\rho_i},
    \qquad i=1,\ldots,N_r.
\end{align}
Upon receiving $s=(i,j)$, the decoder reconstructs
\begin{align}
    \widehat{x}_s
    \defeq \bar\rho_i v_j.
    \label{eq:polar_reconstruction}
\end{align}
This reconstruction guarantees
\begin{align}
    \norm{x-\widehat{x}_s}\leq\delta_0\norm{x},
    \qquad x\in C_s.
    \label{eq:polar_error}
\end{align}
Appendix~\ref{ap:polar_quantizer} derives~\eqref{eq:polar_error} and the
decoded uncertainty ball used below. The quantizer uses $N_rN_a$ symbols and is feasible whenever
\begin{align}
    N_rN_a\leq L=2^R.
\end{align}
\end{definition}

After the first time step, we switch to a uniform quantizer for
a simpler recursive analysis.  For each transformed coordinate, choose an integer
$L_i\geq1$ such that
\begin{align}
    \prod_{i=1}^d L_i\leq L.
    \label{eq:primitive_rate_allocation}
\end{align}
To this end, we use the primitive quantizer defined
in~\cite{tatikonda2004control}, adapted to the present notation.
Without loss of generality, scale the rows $\phi_i^*$ of $\Phi_A$ so that
$\norm{\phi_i}=1$.

\begin{definition}[Primitive quantizer]
\label{def:primitive}
A primitive quantizer is the four-tuple
$(\overline L,\overline h,\Phi_A,p)$, where
$p\in\R^d$ is the center,
$\overline L=(L_1,\ldots,L_d)^*\in\Z_+^d$ is the vector of cell counts,
$\overline h=(h_1,\ldots,h_d)^*$, with $h_i\geq0$, specifies the
transformed halfwidths, and $\Phi_A$ is the coordinate transformation.
\end{definition}

At time $t$, the primitive quantizer partitions the support
\begin{align}
    \Lambda_t
    \defeq
    \left\{x\in\R^d:
    \left|\phi_i^*(x-p_t)\right|\leq h_{t,i},
    i=1,\ldots,d\right\},
    \label{eq:primitive_support}
\end{align}
where $\phi_i^*$ is the $i$th row of $\Phi_A$.  Transformed coordinate
$i$ is divided into $L_i$ intervals of length $2h_{t,i}/L_i$.  Upon
observing $x_t$, the quantizer:
\begin{enumerate}
    \item subtracts $p_t$ from $x_t$;
    \item applies the coordinate transformation $\Phi_A$;
    \item identifies the box containing $\Phi_A(x_t-p_t)$, using a fixed
    priority rule for boundary points; and
    \item transmits the symbol representing that box.
\end{enumerate}
If $h_{t,i}=0$, that transformed coordinate is interpreted as a singleton.
If $\widehat z_t$ is the midpoint of the selected transformed box, the
decoder reconstructs
\begin{align}
    \widehat x_t=p_t+\Phi_A^{-1}\widehat z_t.
    \label{eq:primitive_reconstruction}
\end{align}
Consequently, whenever $x_t\in\Lambda_t$,
\begin{align}
    \left|\phi_i^*(x_t-\widehat x_t)\right|
    \leq\frac{h_{t,i}}{L_i},
    \qquad i=1,\ldots,d.
    \label{eq:primitive_error}
\end{align}
Our construction will ensure that the state never leaves the
support.
Define the contraction factors
\begin{align}
    \theta_i&\defeq\frac{|\lambda_i(A)|}{L_i},
    &\theta&\defeq\max_i\theta_i.
    \label{eq:tracker_contraction}
\end{align}
\begin{theorem}\label{th:achievability}
Suppose that the time-zero quantizer in Definition~\ref{def:log} uses
$N_rN_a\leq L=2^R$ symbols and the subsequent primitive quantizers use cell
counts satisfying~\eqref{eq:primitive_rate_allocation}.  With $M$, $\delta_0$,
and $\theta$ defined in~\eqref{eq:stationary_weights},~\eqref{eq:polar_delta},
and~\eqref{eq:tracker_contraction}, respectively, if $\theta<1$ and
\begin{align}
    \delta_0^2\left[
    \norm{M}+
    \frac{d\norm{M}\norm{\Phi_A^{-1}}^2\theta^2}
    {(1-\delta_0)^2(1-\theta^2)}
    \right]
    <\lambda_{\min}(H_\gamma),
    \label{eq:rate_main}
\end{align}
then the quantization scheme and stationary controller
described in Appendix~\ref{ap:achievability_proof} satisfy
Problem~\ref{prob:quant_hinf}, uniformly over $\Omega_0$, for every fixed
$P_f\succeq0$.  In the special case $P_f=P$, the same bound holds at every
finite horizon $N$.
\end{theorem}
\long\gdef\AchievabilityProof{
\begin{proof}
Consider the following observer and controller design.
\begin{enumerate}
    \item At time $t=0$, apply the logarithmic quantizer in
    Definition~\ref{def:log} directly to $x_0$.  From the transmitted
    symbol $s$, both sides determine $\widehat x_s$.  To initialize the
    subsequent tracker, define
    $c_s\defeq\widehat{x}_s/(1-\delta_0^2)$ and
    $r_s\defeq\delta_0\norm{c_s}$.  By~\eqref{eq:polar_decoded_ball},
    $C_s\subseteq\{x:\norm{x-c_s}\leq r_s\}$.  Set
    $\widehat x_0=\widehat x_s$ and apply $u_0=-K\widehat x_0$.

    \item Choose the first primitive quantizer with center and transformed
    halfwidths
    \begin{align}
        p_1&=Ac_s+B_uu_0,
        &h_{1,i}&=|\lambda_i(A)|r_s.
        \label{eq:primitive_initialization}
    \end{align}

    \item For each $t\geq1$, the encoder sends the index of the box
    containing $\Phi_A(x_t-p_t)$.  If $\widehat z_t$ is the
    corresponding midpoint offset, the decoder reconstructs
    $\widehat x_t=p_t+\Phi_A^{-1}\widehat z_t$ according
    to~\eqref{eq:primitive_reconstruction} and applies $u_t=-K\widehat x_t$.

    \item Update
    \begin{align}
        p_{t+1}&=A\widehat x_t+B_uu_t,
        &h_{t+1,i}&=\theta_i h_{t,i},
        \label{eq:primitive_update}
    \end{align}
    and repeat Step~3.
\end{enumerate}
The quantizer parameters depend only on the received symbols,
past controls, and known system matrices, so the encoder and decoder compute
them synchronously.

Let $\Lambda_t$ denote the support~\eqref{eq:primitive_support} with center
$p_t$.  We first verify that the quantizers do not overload.  From
\eqref{eq:polar_decoded_ball}
and~\eqref{eq:primitive_initialization},
\begin{align}
    \left|\phi_i^*(x_1-p_1)\right|
    &=\left|\lambda_i(A)\phi_i^*(x_0-c_s)\right|\nonumber\\
    &\leq |\lambda_i(A)|r_s=h_{1,i},
\end{align}
so $x_1\in\Lambda_1$.  If $x_t\in\Lambda_t$, then
\eqref{eq:primitive_error} and~\eqref{eq:primitive_update} give
\begin{align}
    \left|\phi_i^*(x_{t+1}-p_{t+1})\right|
    &=\left|\lambda_i(A)\phi_i^*
    (x_t-\widehat x_t)\right|\nonumber\\
    &\leq\frac{|\lambda_i(A)|}{L_i}h_{t,i}
    =h_{t+1,i}.
\end{align}
Thus $x_{t+1}\in\Lambda_{t+1}$, and induction proves the no-overload claim.
Define $e_t\defeq x_t-\widehat x_t$.
The support recursion and
\eqref{eq:primitive_error} imply, for $t\geq1$,
\begin{align}
    \left|\phi_i^*e_t\right|
    \leq\theta_i^t r_s.
    \label{eq:tracking_error_coordinate}
\end{align}
It follows that
\begin{align}
    \sum_{t=1}^{\infty}e_t^*Me_t
    &\leq\norm{M}\sum_{t=1}^{\infty}\norm{e_t}^2\nonumber\\
    &\leq d\norm{M}\norm{\Phi_A^{-1}}^2r_s^2
    \sum_{t=1}^{\infty}\theta^{2t}\nonumber\\
    &=\frac{d\norm{M}\norm{\Phi_A^{-1}}^2\theta^2}
    {1-\theta^2}r_s^2.
    \label{eq:tracking_error_tail}
\end{align}
Moreover,~\eqref{eq:polar_error} gives
\begin{align}
    e_0^*Me_0
    \leq\norm{M}\delta_0^2\norm{x_0}^2.
    \label{eq:initial_error_cost}
\end{align}
Since $x_0\in\{x:\norm{x-c_s}\leq r_s\}$ and
$r_s=\delta_0\norm{c_s}$ by~\eqref{eq:polar_ball_parameters},
\begin{align}
    r_s\leq\frac{\delta_0}{1-\delta_0}\norm{x_0}.
    \label{eq:decoded_radius_bound}
\end{align}
Combining~\eqref{eq:tracking_error_tail}--\eqref{eq:decoded_radius_bound},
we obtain
\begin{align}
    \sum_{t=0}^{\infty}e_t^*Me_t
    \leq\delta_0^2\left[
    \norm{M}+
    \frac{d\norm{M}\norm{\Phi_A^{-1}}^2\theta^2}
    {(1-\delta_0)^2(1-\theta^2)}
    \right]\norm{x_0}^2.
    \label{eq:total_error_cost}
\end{align}

Requiring the coefficient of $\norm{x_0}^2$ in~\eqref{eq:total_error_cost}
to be strictly smaller than $\lambda_{\min}(H_\gamma)$ gives the sufficient
condition~\eqref{eq:rate_main}. Under this condition, the coefficient is at most
$\lambda_{\min}(H_\gamma)-\eta$ for some $\eta>0$.  Hence, uniformly over
$x_0\in\Omega_0$,
\begin{align}
    x_0^*Px_0+\sum_{t=0}^{\infty}e_t^*Me_t
    &\leq\gamma^2x_0^*\Pi_0^{-1}x_0
    -\eta\norm{x_0}^2.
    \label{eq:stationary_cost_margin}
\end{align}
Moreover, the state under the stationary controller satisfies
\begin{align}
    x_{t+1}=(A-B_uK)x_t+B_uKe_t.
\end{align}
Since $P$ is the stabilizing Riccati solution, $A-B_uK$ is Schur.  Hence
there are $C_F<\infty$ and $\rho\in(0,1)$ such that
$\norm{(A-B_uK)^t}\leq C_F\rho^t$.  The geometric error bounds above also
give $\norm{e_t}\leq C_e\theta^t\norm{x_0}$ for some $C_e<\infty$.
Consequently,
\begin{align}
    \frac{\norm{x_t}}{\norm{x_0}}
    &\leq C_F\rho^t+C_F\norm{B_uK}C_e
    \sum_{j=0}^{t-1}\rho^{t-1-j}\theta^j
    \longrightarrow0,
    \label{eq:stationary_state_decay}
\end{align}
uniformly over $x_0\in\Omega_0$.  Thus the terminal mismatch in
\eqref{eq:stationary_cost_identity} satisfies
\begin{align}
    \sup_{x_0\in\Omega_0}
    \frac{\left|x_{N+1}^*(P_f-P)x_{N+1}\right|}
    {x_0^*\Pi_0^{-1}x_0}
    \longrightarrow0.
    \label{eq:terminal_mismatch_decay}
\end{align}
Since $u_t+Kx_t=Ke_t$, combining~\eqref{eq:stationary_cost_identity},
\eqref{eq:stationary_cost_margin}, and
\eqref{eq:terminal_mismatch_decay} gives
\begin{align}
    \limsup_{N\to\infty}\ \sup_{x_0\in\Omega_0}
    \frac{J_N^{P_f}(x_0)}{x_0^*\Pi_0^{-1}x_0}
    &\leq\gamma^2-
    \frac{\eta}{\lambda_{\max}(\Pi_0^{-1})}
    <\gamma^2.
    \label{eq:stationary_asymptotic_performance}
\end{align}
In the particular case $P_f=P$, the terminal mismatch is identically zero,
and~\eqref{eq:stationary_cost_margin} shows that the desired bound holds for
every finite $N$.
\end{proof}
}
\ifachproofinappendix
\emph{Proof sketch.}
We use Definition~\ref{def:log} at $t=0$, then recursively update the
primitive quantizer of Definition~\ref{def:primitive} to prevent overload
and ensure geometric error decay. With $u_t=-K\widehat x_t$, stability and
error decay eliminate the terminal mismatch uniformly (see
Appendix~\ref{ap:achievability_proof}).
\else
\AchievabilityProof
\fi

\begin{remark}[Limiting rate regimes]
\label{rem:rate_allocation}
Under the continuous relaxation $R_i=\log_2L_i$, the optimal allocation for
$\theta\in(0,1)$ satisfies
\begin{align}
    \sum_{i=1}^d
    \left[\log_2\frac{|\lambda_i(A)|}{\theta}\right]_+=R.
    \label{eq:tracking_rate_allocation}
\end{align}
As $\gamma\to\infty$,~\eqref{eq:rate_main} becomes inactive.  Letting
$\theta\uparrow1$ then recovers the data-rate condition
\begin{align}
    R>\sum_{i=1}^d\left[\log_2|\lambda_i(A)|\right]_+,
    \label{eq:stabilization_recovered}
\end{align}
which is~\eqref{eq:drt}.
In the high-performance limit $\lambda_{\min}(H_\gamma)\downarrow0$,
Appendix~\ref{ap:high_rate_allocation} gives, for $\norm{M}>0$, the sufficient rate threshold
\begin{align}
    R_{\rm ach}(\gamma)
    &=\frac d2\log_2
    \frac{\norm{M}}{\lambda_{\min}(H_\gamma)}+O(1).
    \label{eq:vector_high_rate_scaling}
\end{align}
For $d>1$, this is an achievable scaling; for $d=1$, it matches the converse
to leading order, as shown next.
\end{remark}
See Appendix~\ref{ap:high_rate_allocation} for full details.
\begin{remark}[Scalar specialization]\label{rem:scalar_ach}
Suppose $d=1$, $K\neq0$, and $\Phi_A=1$, and use $\Bg$ from
\eqref{eq:def_Bgamma_converse}.
For an even alphabet size $L$, use $N_a=2$, $N_r=L/2$ at time zero and
$L_1=L$ thereafter, and set
\begin{align}
    \delta_L&\defeq\tanh\left(\frac{\log(\ub/\lb)}{L}\right),
    &\theta_L&\defeq\frac{|A|}{L}.
    \label{eq:scalar_delta_theta}
\end{align}
With $R=\log_2L$, Theorem~\ref{th:achievability} reduces to
\begin{align}
    L&>|A|,\nonumber\\
    \delta_L^2\left[
    1+\frac{\theta_L^2}
    {(1-\delta_L)^2(1-\theta_L^2)}
    \right]&<\Bg.
    \label{eq:ach_scalar}
\end{align}
Let $L_{\rm ach}$ be the smallest even integer satisfying~\eqref{eq:ach_scalar}.
As $\Bg\downarrow0$, $L_{\rm ach}\to\infty$, so $\theta_{L_{\rm ach}}\to0$.
Using $\delta_L$ from~\eqref{eq:scalar_delta_theta} and $\tanh z=z+o(z)$,
we solve~\eqref{eq:ach_scalar} for $L$ and take logarithms to obtain
\begin{align}
    \log_2L_{\rm ach}
    &=\log_2\left(\frac{\log(\ub/\lb)}{\sqrt{\Bg}}\right)+o(1).
    \label{eq:ach_high_rate}
\end{align}
Rounding up to the next even integer changes the rate by only $o(1)$.
Hence the scalar logarithmic scheme is first-order rate-optimal in the
high-rate limit.
\end{remark}

\section{Simulations} \label{sec:num}
We consider the scalar plant $x_{t+1}=4x_t+u_t$ with
$\Omega_0=[-2,-0.1]\cup[0.1,2]$, $Q=\Pi_0=0.1$, and $W=1$.
The controller uses the stationary gain $K$ and
terminal weight $P_f=P$.  We sample even alphabet sizes $L\geq6$ and report
the rate $\log_2L$ bits per sample.
The proposed and uniform worst-case gains at $N=50$ are computed exactly;
for Fu--Xie, we maximize
$\sqrt{J_{50}^{P}(x_0)/(x_0^*\Pi_0^{-1}x_0)}$ over a dense grid in
$\Omega_0$.

Figure~\ref{fig:num_simulation} compares the proposed scheme with a
finite-level logarithmic quantizer with dynamic scaling adapted from
\cite{fu2009finite}, whose logarithmic spacing, zoom factors,
and initial scale are tuned numerically subject to a sufficient stability
condition, and
a uniform-initialization/primitive-tracking scheme adapted from
\cite{tatikonda2004control}.
The proposed curve closely follows the sufficient bound.
Appendix~\ref{ap:dynamic_range_sweep} reports an additional
sweep over the initial-set dynamic range.

\begin{figure}[htbp]
    \centering
    \includegraphics[width=0.95\linewidth]{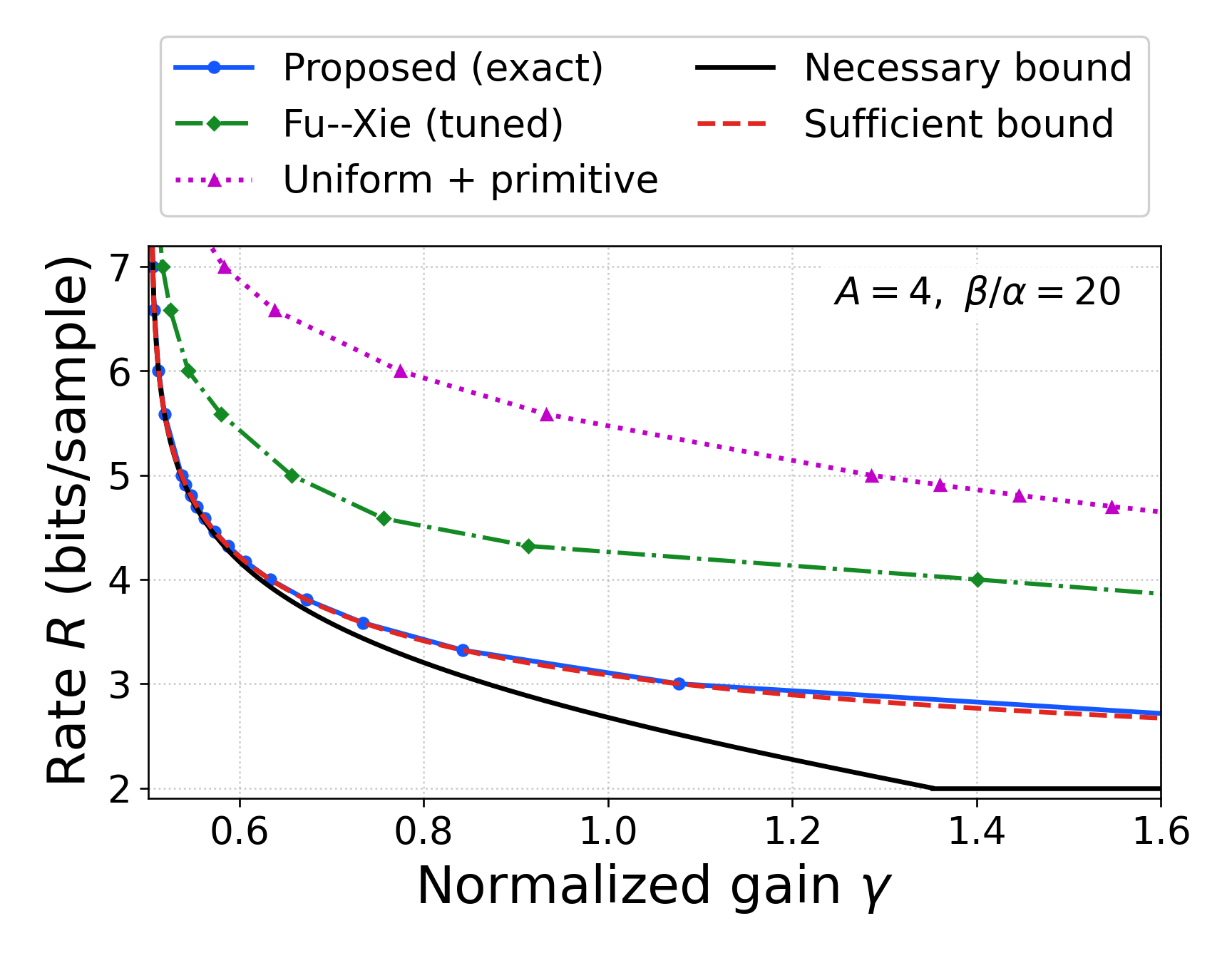}
    \caption{Rate versus normalized gain for the stationary scalar example
    with $A=4$ and $\beta/\alpha=20$.
    Markers show worst-case gains at $N=50$, computed exactly
    for the proposed and uniform schemes and estimated by sampling for
    Fu--Xie. The black curve is the combined
    necessary bound
    in~\eqref{eq:combined_converse}; the dashed red curve is the continuous
    relaxation of the sufficient bound in~\eqref{eq:ach_scalar}.  Each marker
    uses an actual alphabet of size $L$ and is plotted at rate $\log_2L$.}
    \label{fig:num_simulation}
\end{figure}

\section{Conclusion}
This paper analyzes $\Hinf$ control under a finite data rate communication constraint in the
noiseless setting. We establish a lower bound on the required data rate for achieving a specified
$\Hinf$ performance in the scalar case using the time-zero covering argument.
The achievability scheme we developed using a logarithmic quantizer
matches the scalar converse asymptotically in the high-data-rate
limit.
The scalar converse uses only the time-zero cost and may be loose
at moderate rates, while the vector achievability bound uses worst-case
matrix norms and may be conservative for anisotropic systems. The analysis
also assumes exact state observations, a noiseless plant and channel, a
bounded annular initial-state set, and real-diagonalizable $A$.
For vector systems, a matching converse would require
anisotropic radial--angular covering bounds, suggesting a logarithmic
radial quantizer paired with a near-uniform angular quantizer. Process or
measurement noise would also require uncertainty propagation or joint
estimation and quantization.

\bibliographystyle{IEEEtran}

\clearpage
\section*{Appendix}
\appendices

\section{Derivation of the Time-Zero Polar Quantizer}
\label{ap:polar_quantizer}
The radial quantizer is logarithmic.  Thus, with
$h\defeq T/(2N_r)$, adjacent radial boundaries satisfy
\begin{align}
    \frac{\rho_i}{\rho_{i-1}}=e^{2h}.
\end{align}
Fix a symbol $s=(i,j)$ and a point $x\in C_s$.  Write
\begin{align}
    x=\rho u,
    \qquad
    \rho=\norm{x}\in[\rho_{i-1},\rho_i],
    \qquad
    u\in S^{d-1}.
\end{align}
The encoder represents the direction $u$ by $v_j$, where
\begin{align}
    \vartheta\defeq\angle(u,v_j)\leq\Theta.
\end{align}
Let $\bar\rho_i>0$ denote the reconstruction radius for shell $i$.
Upon receiving $s=(i,j)$, the decoder reconstructs
\begin{align}
    \widehat{x}_s\defeq\bar\rho_i v_j.
\end{align}
We now choose $\bar\rho_i$ to minimize the largest relative error in the
cell.  For $x=\rho u$, the squared relative error is
\begin{align}
    \mathcal E(\rho,\vartheta;\bar\rho_i)
    &\defeq\frac{\norm{x-\widehat{x}_s}^2}{\norm{x}^2} \nonumber\\
    &=1+\left(\frac{\bar\rho_i}{\rho}\right)^2
      -2\frac{\bar\rho_i}{\rho}\cos(\vartheta).
\end{align}
This expression increases with $\vartheta$ and is convex in $1/\rho$.
Consequently, its worst-case value over the closed sector containing
$C_s$ occurs at $\vartheta=\Theta$ and at one of the radial endpoints.
The design problem therefore reduces to
\begin{align}
    \min_{\bar\rho_i>0}\max\left\{
      \mathcal E(\rho_{i-1},\Theta;\bar\rho_i),
      \mathcal E(\rho_i,\Theta;\bar\rho_i)
    \right\}.
\end{align}
The minimax choice equates the two endpoint errors which gives
\begin{align}
    1+\frac{\bar\rho_i^2}{\rho_{i-1}^2}
      -\frac{2\bar\rho_i\cos(\Theta)}{\rho_{i-1}}
    =1+\frac{\bar\rho_i^2}{\rho_i^2}
      -\frac{2\bar\rho_i\cos(\Theta)}{\rho_i}.
\end{align}
The positive solution for this is
\begin{align}
    \bar\rho_i
    &=\frac{2\rho_{i-1}\rho_i}{\rho_{i-1}+\rho_i}\cos(\Theta)
    \nonumber\\
    &=\cos(\Theta)\operatorname{sech}(h)
      \sqrt{\rho_{i-1}\rho_i} \nonumber\\
    &=\sqrt{1-\delta_0^2}\sqrt{\rho_{i-1}\rho_i}.
    \label{eq:polar_optimal_radius}
\end{align}
The second equality uses the logarithmic spacing, and the third follows
from \eqref{eq:polar_delta}.  To see that it is minimax, note that
the endpoint error at $\rho_k$, $k\in\{i-1,i\}$, is minimized at
$\bar\rho_i=\rho_k\cos(\Theta)$, while
\begin{align}
    \rho_{i-1}\cos(\Theta)
    <\bar\rho_i
    <\rho_i\cos(\Theta).
\end{align}
Thus the outer-endpoint error decreases up to the balancing point, and
the inner-endpoint error increases after it.

It remains to evaluate its worst-case error.  For every $x\in C_s$,
\begin{align}
    \frac{\norm{x-\widehat{x}_s}^2}{\norm{x}^2}
    &\leq \mathcal E(\rho_{i-1},\Theta;\bar\rho_i) \nonumber\\
    &=1-\cos^2(\Theta)\operatorname{sech}^2(h) \nonumber\\
    &=\sin^2(\Theta)+\tanh^2(h)\cos^2(\Theta) \nonumber\\
    &=\delta_0^2.
\end{align}
Taking square roots proves \eqref{eq:polar_error}.

We next construct the decoded uncertainty ball used to initialize the
subsequent primitive quantizer.  Define the relative-error region
\begin{align}
    \mathcal D_s
    \defeq
    \{y\in\R^d:\norm{y-\widehat{x}_s}\leq\delta_0\norm{y}\}.
\end{align}
Since $\delta_0<1$, squaring the inequality defining $\mathcal D_s$ and
expanding gives
\begin{align}
    (1-\delta_0^2)\norm{y}^2
    -2y^\ast\widehat{x}_s+\norm{\widehat{x}_s}^2\leq0.
\end{align}
Dividing by $1-\delta_0^2$ and completing the square give
\begin{align}
    \norm{y-\frac{\widehat{x}_s}{1-\delta_0^2}}^2
    \leq
    \delta_0^2
    \norm{\frac{\widehat{x}_s}{1-\delta_0^2}}^2.
\end{align}
Define
\begin{align}
    c_s&\defeq\frac{\widehat{x}_s}{1-\delta_0^2},
    &r_s&\defeq\delta_0\norm{c_s}.
    \label{eq:polar_ball_parameters}
\end{align}
By~\eqref{eq:polar_error}, the decoded cell satisfies
\begin{align}
    C_s\subseteq\mathcal D_s
    &=\{y\in\R^d:\norm{y-c_s}\leq r_s\}.
    \label{eq:polar_decoded_ball}
\end{align}

When $d=1$, $\Theta=0$, $v_j\in\{-1,+1\}$, and
$\delta_0=\tanh(h)$.  On a positive cell,
\begin{align}
    \widehat{x}_i=\bar\rho_i
    &=\frac{2\rho_{i-1}\rho_i}{\rho_{i-1}+\rho_i}, \nonumber\\
    c_i&=\frac{\widehat{x}_i}{1-\delta_0^2}
        =\frac{\rho_{i-1}+\rho_i}{2}, \nonumber\\
    r_i&=\delta_0c_i
        =\frac{\rho_i-\rho_{i-1}}{2}.
\end{align}
The negative-cell reconstruction and center follow by reflection, while
the radius remains positive.  If all $L=2N_r$ symbols are used, then
$h=T/L$ and $\delta_0=\tanh(T/L)$.

\ifachproofinappendix
\section{Proof of Theorem~\ref{th:achievability}}
\label{ap:achievability_proof}
\AchievabilityProof
\fi

\section{Rate Allocation in the Limiting Regimes}
\label{ap:high_rate_allocation}
We first derive~\eqref{eq:tracking_rate_allocation}.  Relax the integer
constraint on $L_i$ and write $R_i=\log_2L_i\geq0$.  For any prescribed
$\theta>0$, the condition
\begin{align}
    \max_i |\lambda_i(A)|2^{-R_i}\leq \theta
\end{align}
holds if and only if
\begin{align}
    R_i\geq
    \left[\log_2\frac{|\lambda_i(A)|}{\theta}\right]_+,
    \qquad i=1,\ldots,d.
\end{align}
The minimum total rate for this value of $\theta$ is therefore
\begin{align}
    \sum_{i=1}^d R_i
    =\sum_{i=1}^d
    \left[\log_2\frac{|\lambda_i(A)|}{\theta}\right]_+.
\end{align}
This proves~\eqref{eq:tracking_rate_allocation}.  Letting $\theta\uparrow1$
gives the strict data-rate
threshold~\eqref{eq:stabilization_recovered}.  If $A$ is nonsingular and
$R$ is large enough that all coordinates are active, then
\begin{align}
    R
    &=\sum_{i=1}^d\log_2|\lambda_i(A)|-d\log_2\theta,
\end{align}
and hence
\begin{align}
    \theta=|\det A|^{1/d}2^{-R/d}.
    \label{eq:appendix_theta_high_rate}
\end{align}

We next analyze the time-zero polar allocation. For fixed $d\geq2$, choose
the angular codebook to be a cardinality-minimal geodesic $\Theta$-cover of
$S^{d-1}$. For $d\geq3$, the equal-cap covering-density bound and cap-measure
estimates in~\cite[Thm.~2.2 and Lem.~3.4]{naszodi_2016_covering} imply the
following scaling. For $d=2$, the same conclusion follows from the exact circle
covering number $N_a=\lceil\pi/\Theta\rceil$. Thus, there exist constants
$0<c_d\leq C_d<\infty$ such that
\begin{align}
    c_d\Theta^{-(d-1)}
    \leq N_a\leq
    C_d\Theta^{-(d-1)},
    \qquad \Theta\downarrow0,
    \label{eq:spherical_cover_scaling}
\end{align}

Let $L=2^R$ and choose $N_r=\lfloor L/N_a\rfloor$. This uses at most $L$
symbols. For sufficiently large $L$, the allocation below
and~\eqref{eq:spherical_cover_scaling} give
\begin{align}
    \frac{L\Theta^{d-1}}{2C_d}
    \leq N_r\leq\frac{L\Theta^{d-1}}{c_d}.
    \label{eq:radial_count_bounds}
\end{align}
With $h=T/(2N_r)$,~\eqref{eq:polar_delta} gives, for sufficiently small
$\Theta$ and $h$,
\begin{align}
    \frac{\Theta^2+h^2}{2}
    \leq\delta_0^2\leq\Theta^2+h^2.
    \label{eq:polar_high_rate_balance}
\end{align}
Balancing the radial and angular terms, choose $\Theta=L^{-1/d}$.
Equations~\eqref{eq:spherical_cover_scaling} and~\eqref{eq:radial_count_bounds}
then give
\begin{align}
    \Theta&=L^{-1/d},\nonumber\\
    \frac{L^{1/d}}{2C_d}&\leq N_r\leq\frac{L^{1/d}}{c_d},\nonumber\\
    c_dL^{(d-1)/d}&\leq N_a\leq C_dL^{(d-1)/d}.
    \label{eq:polar_high_rate_allocation}
\end{align}
As $R\to\infty$, $L\to\infty$, so~\eqref{eq:polar_high_rate_allocation}
gives $\Theta\to0$ and $N_r\to\infty$. Hence $h=T/(2N_r)\to0$,
and~\eqref{eq:polar_delta} implies $\delta_0\to0$.
The bounds also show that $\delta_0^2$ lies between fixed positive
multiples of $L^{-2/d}=2^{-2R/d}$.
To obtain~\eqref{eq:vector_high_rate_scaling}, consider the high-performance
limit $\lambda_{\min}(H_\gamma)\downarrow0$ with $\norm{M}>0$.
Condition~\eqref{eq:rate_main} forces $\delta_0\to0$, hence $R\to\infty$;
under~\eqref{eq:appendix_theta_high_rate}, $\theta\to0$ as well.
In this high-rate regime, the left-hand side of~\eqref{eq:rate_main}
satisfies
\begin{align}
    \delta_0^2\left[
    \norm{M}+
    \frac{d\norm{M}\norm{\Phi_A^{-1}}^2\theta^2}
    {(1-\delta_0)^2(1-\theta^2)}
    \right]
    &=\norm{M}\delta_0^2\left(1+O(\theta^2)\right).
\end{align}
Taking base-2 logarithms of the bounds obtained from~\eqref{eq:rate_main}
using~\eqref{eq:polar_high_rate_balance} and~\eqref{eq:polar_high_rate_allocation}
gives the rate threshold~\eqref{eq:vector_high_rate_scaling} in this limit: fixed multiplicative
constants contribute only to the additive $O(1)$ term.
For any fixed $H_\gamma\succ0$,~\eqref{eq:rate_main} holds for sufficiently large $R$.
The time-zero error therefore determines the leading term of the high-rate bound.

\section{Additional Numerical Study}
\label{ap:dynamic_range_sweep}
We repeat the comparison at $A=4$, $L=12$, and $N=50$, keeping $\beta=2$
and varying $\kappa_0\defeq\beta/\alpha$ from $2$ to $100$. All cost weights
match Figure~\ref{fig:num_simulation}, with the corresponding stationary
gain $K$ and terminal weight $P_f=P$. The Fu--Xie logarithmic spacing,
zoom factors, and initial scale are retuned at each $\kappa_0$, using the
same stability checks and dense initial-state sampling as in
Figure~\ref{fig:num_simulation}. Figure~\ref{fig:dynamic_range_sweep} shows
that the proposed scheme achieves a lower gain than both comparison
schemes throughout the tested range.

\begin{center}
\begin{minipage}{0.92\linewidth}
    \centering
    \includegraphics[width=\linewidth]{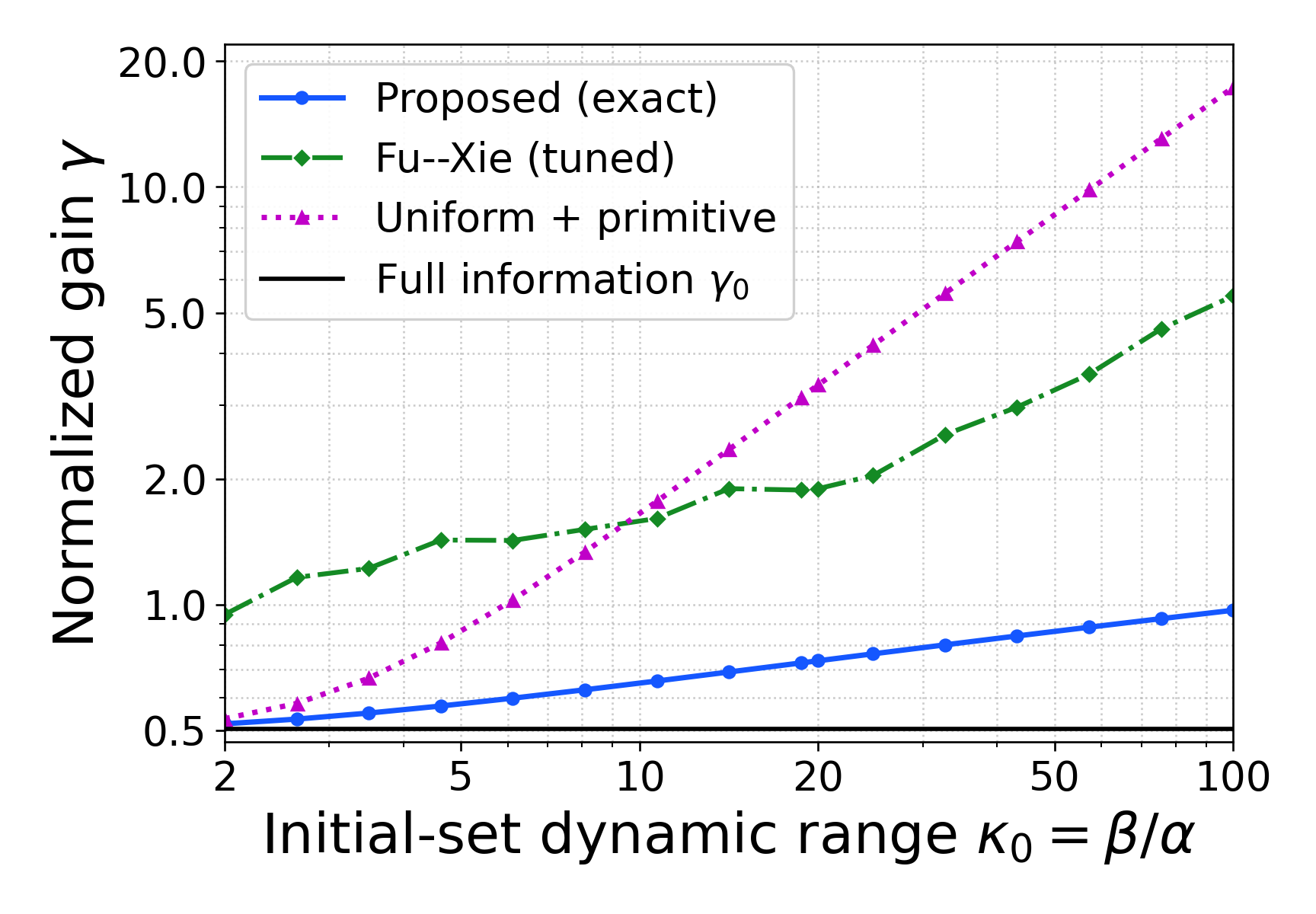}
    \captionsetup{type=figure}
    \captionof{figure}{Normalized gain versus the initial-set dynamic range at
    $L=12$ ($R\approx3.585$ bits/sample) and $N=50$. Gains are computed exactly for the
    proposed and uniform schemes and estimated by sampling for the tuned,
    stability-screened Fu--Xie designs. The full-information value
    $\gamma_0$ is shown for reference.}
    \label{fig:dynamic_range_sweep}
\end{minipage}
\end{center}

\end{document}